\documentclass[11pt]{article} 

\usepackage{graphics,graphicx,varwidth}
\usepackage[dvips]{epsfig}
\usepackage{eepic}
\usepackage{color}
\usepackage{amsmath,amssymb,amsfonts}
\usepackage[lined]{algorithm2e}
\usepackage{amsthm}
\usepackage[english]{babel}
\usepackage{cite}
\newtheorem{theorem}{Theorem}[section]
\newtheorem{proposition}[theorem]{Proposition}
\newtheorem{definition}[theorem]{Definition}
\newtheorem{lemma}[theorem]{Lemma}

\newtheorem{corollary}[theorem]{Corollary}
\newtheorem{fact}[theorem]{Fact}
\newtheorem{remark}[theorem]{Remark}

\usepackage[small,it]{caption}
\usepackage{mymacros}
\usepackage{subfigure}
\usepackage{authblk}
\usepackage{paralist}
\usepackage{fullpage,times}
\usepackage{color}
\usepackage[lined]{algorithm2e}

\usepackage{tikz}
\usetikzlibrary{arrows,automata}
\usetikzlibrary{positioning}
\tikzset{
    state/.style={
           rectangle,
           rounded corners,
           draw=black, very thick,
           minimum height=2em,
           fill=green!80!magenta!50,
           inner sep=2pt,
           text centered,
           },
   trans/.style={
           rectangle,
           rounded corners,
           draw=black, very thick,
           minimum height=2em,
           fill=cyan,
           inner sep=2pt,
           text centered,
           },
   middle/.style={
           rectangle,
           rounded corners,
           draw=white, very thick,
           minimum height=2em,
           fill=white,
           inner sep=2pt,
           text centered,
           }
}

\definecolor{lightgray}{gray}{0.5}
\usepackage{amsmath,amssymb,amsthm} 

\renewcommand\epsilon\varepsilon 
\renewcommand\phi\varphi 

\newcommand\NN{\mathbb{N}} 
\newcommand\RR{\mathbb{R}} 
\newcommand\CC{\mathbb{C}} 

\newcommand\rounds{R} 

\newcommand\bits{T} 

\newcommand{\sd}{\textbf{int}}

\newcommand{\stt}{\sigma}

\renewcommand{\CC}{\emph{CC}}
\newcommand{\SC}{SC }

\newcommand{\co}{}
\newcommand{\str}{\mathcal{S}}
\newcommand{\intstr}{\mathcal{S}^{\sd}}
\newcommand{\bigO}{O}

\newcommand{\algo}{\textbf{Alg}}
\newcommand{\cutl}{W}
\newcommand{\cutr}{\overline{W}}
\newcommand{\ravg}{C^{avg}}
\newcommand{\AND}{AND}
\newcommand{\OR}{OR}
\newcommand{\embedcost}{A_{\ell +1}}

\usepackage{comment}

\begin{document}
\date{}
\title{Streaming Communication Protocols\thanks{This work has been partially supported by 
the ERC project QCC and
    the French ANR Blanc project RDAM}} 

\author{Lucas Boczkowski\thanks{E-mail: {\tt lucas.boczkowski@liafa.univ-paris-diderot.fr}}}
\author{	Iordanis Kerenidis\thanks{E-mail:  {\tt iordanis.kerenidis@liafa.univ-paris-diderot.fr}}}
\author{
   Fr\'ed\'eric Magniez\thanks{E-mail: {\tt frederic.magniez@cnrs.fr}}
}
\affil{CNRS, IRIF, Univ Paris Diderot, Paris, France.}
\begin{titlepage}
\maketitle
\date{}

\def\thefootnote{\fnsymbol{footnote}}

\thispagestyle{empty}

\begin{abstract}
We define the Streaming Communication model that combines the main aspects of communication complexity and streaming. We consider two agents that want to compute some function that depends on inputs that are distributed to each agent. The inputs arrive as data streams and each agent has a bounded memory. Agents are allowed to communicate with each other and also update their memory based on the input bit they read and the previous message they received. 

We provide tight tradeoffs between the necessary resources, i.e. communication and memory, for some of the canonical problems from communication complexity by proving a strong general lower bound technique. Second, we analyze the Approximate Matching problem and show that the complexity of this problem (i.e. the achievable approximation ratio) in the one-way variant of our model is strictly different both from the streaming complexity and the one-way communication complexity thereof. 
\end{abstract}
\end{titlepage}


\section{Introduction}
In the last decade we have witnessed a big shift in the way data is produced and computation is performed.
First, we now have to deal with enormous amounts of data that we cannot
even store in memory (internet traffic, CERN experiments, space expeditions). 
Second, computations do not happen in a single processor or machine, but with multi-core processors
and multiple machines in cloud architectures.
All these real-world changes necessitate that we revisit and
extend our models and tools for studying the efficiency and hardness of
computational problems.

Imagine the following situation : some input is spread among two or more agents. The agents
want to compute some function $f$ which depends on everybody's input.
This is an archetypal problem of Communication Complexity (CC) \cite{Yao79}, which offers a way to estimate the number of bits that need to be exchanged, under various settings, in order to achieve that goal. 
There are many different CC models, depending whether the agents can speak directly between them or through a referee, and whether they can use multiple rounds of communication or just a single one. Communication complexity has found a variety of applications both in networks and distributed computing but also in other areas of theoretical computer science, including, circuit lower bounds, fomulae size, VLSI design, etc. 
All these communication models however have a feature in common. They assume the agents are computationally unbounded and that the input is delivered all at once. 

In a distributed context as that of sensor networks, not only are there several 
computing agents, the input might not even be given all at once.
The \emph{streaming model} \cite{ams} has been defined precisely to capture the fact that the input of an agent is so big it cannot be stored or read several times. Instead it comes bit by bit. Some function of the stream needs to be computed, but the available space is not big enough to store the entire input. The streaming model has been extensively studied in recent years with a plethora of interesting upper and lower bounds on the necessary memory to solve specific streaming problems \cite{mutu}. 
More recently, the turnstile model has received a lot of attention.
In this model, streams are made of both insertions and deletions,
and the function to be computed depends on the remaining elements (and eventually their respective frequencies).
Indeed, any streaming algorithm in the turnstile model can be turned into an algorithm based solely on the updates of linear sketches~\cite{Li14}.

\paragraph{The Streaming Communication model.}
We would like to combine the two above mentioned models to include both that inputs are distributed among different agents and also are 
coming as streams at each agent. Each agent is given a bounded memory to store what she sees.
We refer to this extension as the \emph{Streaming Communication (SC) model}.
Even though communication arguments have often been invoked in proving lower bounds for regular streaming models,
as in the seminal work of \cite{ams,BabaiFS86}, this model has not been rigorously defined previously,
in spite of its theoretical appeal and relevance for actual communication networks.
More formally, in the SC model consider two agents, Alice and Bob, want to compute some function $f$
that depends on inputs $(x,y)$ that are respectively distributed to each agent, $x$ to Alice and $y$ to Bob.
Both inputs arrive as data streams and each agent has a bounded memory of a given size $S$. 
Agents may or may not speak every time they receive a bit. They can also update their memory based on the previous bit they read, the previous message they received and of course the actual content of their memory. 

Additionally to the memory size $S$ of Alice and Bob, the other relevant parameters we consider
are the number $R$ of communication rounds and the number $T$ of bits in the full transcript (the concatenation of all messages).
In the {\em one-way} SC model, there is only a single message from Alice to Bob at the end of the streams.
Observe that we do not bound the size of each message, since we show that those can always be assumed to be of at most $S+1$ bits
(\textbf{Proposition~\ref{proposition:size}}).

\paragraph{Related models}
Before we present our results in the streaming communication model,
we review some related works in the communication and streaming models.
As explained previously, our goal is to provide rigorous tradeoffs between the
two resources: memory storage and communication between the players, in a model where
inputs are coming as streams.

The most relevant works to ours are the papers\cite{Gib1,Gib2}. There, two parties receive two streams and at the end of the streams each party sends their workspace to a referee which uses both workspaces to compute some function of the union of the two streams. In this model, we can also see elements both from streaming algorithms and communication complexity, albeit of the restricted form of simultaneous message passing. Here we provide a more general framework for communication and we look at a much wider variety of problems and protocols.

One of the powerful and often used techniques in streaming algorithms is
linear sketches, which naturally provide very simple SC protocols even in the one-way setting, by just combining the linear sketches at the end of the protocol. In fact, in the turnstile model, when streams are made of sequences of insertions and deletions and the function to be computed is a function of the remaining elements (and eventually their respective frequencies),
any streaming algorithm can be turned into an algorithm based solely on the updates of linear sketches~\cite{Li14}. 
Hence, here we focus on other stream models, such as the one of insertion only, which are more challenging in the context of SC protocols.

%
%

In the distributed functional monitoring setup initially proposed by \cite{CormodeMY08}, $k$ servers receiving a stream have to allow a coordinator to continuously monitor a given quantity. Several works have expanded the results on this model (see e.g. \cite{Chan11,LiuRV12,Huang12, Woodruff12}). These all focus on communication and do not consider
\emph{both} resources,
memory storage and communication simultaneously.   They can be viewed as extending \cite{Gib1,Gib2} with greater number of players. The communication model however is still restricted to a referee scenario.

Another line of work studies bounded-memory versions of communication \cite{LamTT92,BeameTY90,Brody13}. 
Several models have been proposed that share the same structure. 
The input is given all at once, but the players only have
bounded space to store the conversation while further restrictions can be placed
on the algorithms used by the players, for example to be straight line programs \cite{LamTT92}, or branching programs \cite{BeameTY90}.

\paragraph{Our results.}

%


Our first results, detailed in Section \ref{sec:generalities}, show connections between our new model and its two parent models, Communication Complexity and Streaming. 
We show that the total transcript size $T$ and the product $RS$ (of the number of rounds and the memory size)
are both lower bounded by the communication complexity $C(f)$ of the function $f$ we wish to compute, up to logarithmic factors
(\textbf{Proposition~\ref{prop:clocksimth}}).
Those factors come from an inherent notion of clock in our SC model.

The comparison with streaming algorithms is more subtle. Since the $i$-th bit of Alice's input arrives at the same time as the $i$-th bit of Bob's input, the correct comparison is with a single streaming model where the stream is the one we get by interleaving Alice's and Bob's stream. 
We denote by $\intstr(f)$ the memory required by a streaming algorithm for computing a function $f$, when the two streams are interleaved in a single stream. We first observe than interleaving streams instead of concatenating them can lead to an exponential blow up (\textbf{Theorem~\ref{thm:example2}}).
Then we show that $\intstr(f)$ is a lower bound on twice the memory size $S$ of players (\textbf{Proposition~\ref{summem}}).

Then it is natural to ask if there is always a polynomial relation between, on one hand, the parameters of a protocol in the SC model ($S,R$ and $T$), and on the other hand, the communication complexity $C(f)$ (randomized or deterministic) of the function when the input is all given in the beginning and the memory $\intstr(f)$ necessary in the single stream model.
We show that this is not true in general by providing an example for which $\intstr(f)=C(f)= O(\log n)$ 
but $R \cdot S = \Omega(n)$, when $S=\Omega(\log n)$ (\textbf{Theorem~\ref{thm:example}}).
This implies that the SC complexity of a function $f$ may not be immediately derived neither by its communication complexity nor by its streaming complexity. This is one of the main reasons why our model is interesting and necessitates novel techniques for its study.  \medskip

The first of our two main results is a general technique for proving {tradeoffs} between memory and communication in our model. The smaller the memory, the more frequent communication has to be. For instance, one expects that for functions whose communication complexity is $n$, i.e. where all bits are necessary, players with a memory of size $S$ \emph{have to} speak at least every $S$ rounds (either deterministic or randomized), since if they remain silent for more than $S$ rounds, they start to lose information about their input.
More precisely,
assume any function $f$ that can be written as $f(x,y)=G(g_1(x^1,y^1),g_2(x^2,y^2),\ldots,g_L(x^L,y^L)),$
 where $G$ is a function satisfying some assumptions. Then, any randomized protocol computing $f$ must have
 $R\cdot S = \Omega( \sum_{\ell \in L} C(g_\ell) )$ (\textbf{Theorem~\ref{thm:genlem}}).
We can apply our theorem to many canonical communication functions, including $IP_n, DISJ_n$ or $TRIBES_n$, and show that any protocol satisfies $ R \cdot S =\Omega(n)$ (\textbf{Theorem~\ref{thm:corgenlem}}).\medskip

%
In Section \ref{sec:matching}, we study problems arising in the context of graph streaming. 
We work in the insert only model, meaning that the graph is presented as a stream of its edges in an arbitrary order.
Indeed, as opposed to the turnstile model, where any algorithm can be turned into a linear sketches based on~\cite{Li14},
the situation is much more intriguing for problems where linear sketches are not used.
In particular, in the context of streaming algorithms for graph problems, \emph{Approximate Matching} has been extensively studied, and its streaming complexity is still unknown. Given a stream of edges (in an arbitrary order) of an $n$-vertex graph $G$
and some space restriction, the goal is to output a collection of edges from $G$ forming a matching, as big as possible in $G$. The matching size estimation is a different and somehow easier problem~\cite{KKS14}.

It is known that any streaming algorithm for Approximate Matching using $\widetilde{O}(n)$ memory cannot achieve a ratio better than $\frac{e}{e-1}$~\cite{Kapralov13}, whereas
the best known algorithm is a simple greedy algorithm which provides a $2$-approximation. In the one-way CC model, without memory constraints, it has been also showed that a $\frac{3}{2}$-approximation is the tight bound when Alice's message is restricted to $\widetilde{O}(n)$ bits~\cite{Goel12}.
Both these works use in a clever way the so-called Ruzsa-Szemer\'edi graphs. 

We study both the general SC model and its one-way variant. 
Our main bounds are for the one-way variant, the weaker model combining the restrictions of both CC and streaming models:
we show a lower bound of $\frac{e+1}{e-1} \approx 2.16$ for the approximation factor unless $S=n^{1+ \Omega(\frac{1}{\log \log n})}$
(\textbf{Corollary~\ref{cor:lowerb}}),
which is strictly higher than both the single stream lower bound of $\frac{e}{e-1} \approx 1.58$  with same space constraints and the one-way communication lower bound of $1.5$. We also provide a one-way SC protocol achieving an approximation ratio of $3$ with the same space constraints
(\textbf{Theorem~\ref{thm:greedy}}), thus leaving as an open question the optimal approximation ratio.
Moreover, we show that how often the players communicate makes a big difference, namely we show how to implement the simple greedy algorithm when Alice and Bob can communicate during the protocol that provides a ratio of $2$, strictly better than our lower bound for the one-way SC model.


Let us emphasize that all previous lower bounds, including the ones in the turnstile models~\cite{AssadiKLY15,Konrad15}, do not readily apply to the one-way SC model for the Approximate Matching problem.
However, our main lower bound in the one-way SC model uses as a black-box the hard distributions of graph streams of~\cite{Goel12, Kapralov13}.
Therefore, further improvements in the streaming context may lead to improvements in our model.
Given a hard distribution $\mu$ of graphs for the approximate matching for streaming algorithms,
we show how to extend this distribution to produce a hard distribution $\mu_2$ in our one-way SC model (\textbf{Theorem~\ref{thm:reduction}}).

\section{The streaming communication model}\label{sec:prelim}

We provide some background on communication complexity and streaming and then, we define our model and describe some initial results.

\subsection{Communication Complexity}\label{sec:ccback}

We start by reviewing some results in the usual models of communication complexity (CC), defined by Yao \cite{Yao79}. For more details about the communication complexity model, please refer to \cite{NK}. In the communication complexity models, generically denoted by \CC, two players aim at computing some function which depends on their disjoint inputs, by communicating. Each player determines her message based on previous messages and her input. The goal is to minimize the total length of the protocol transcript.


In the randomized case, we will allow the players to share public randomness. Allowing for public randomness makes our lower bounds stronger, while the protocols we provide will be deterministic. We will also consider the expected, rather than maximal, length of transcripts and define the average randomized communication complexity of a function.

\begin{definition}
For a given protocol $\Pi$, we denote by
 $\Pi(x,y,r)$ the transcript with inputs $x,y$ and public randomness $r$.
The worst case (resp. expected) communication complexity of a function $f$
with error $\epsilon$ is defined as
$C_{\varepsilon}(f) = \min_{\Pi } \max_{x,y} \max_r |\Pi(x,y, r)|$ and
$C^{avg}_{\varepsilon}(f) = \min_{\Pi } \max_{x,y} \mathbb{E}_r(|\Pi(x,y, r)|)$,
where the minimum is taken over protocols computing $f$ with error $\varepsilon$, and the expectation on the second line is with
respect to the randomness $r$ used in $\Pi$. 
\end{definition}

The following proposition relates the average and worst case randomized communication complexities.

\begin{proposition}[\cite{NK}]
For any $\epsilon, \delta>0$, it holds that, $
\delta \cdot C_{\varepsilon+\delta}(f) \leq \ravg_{\varepsilon}(f) \leq C_{\varepsilon}(f).$
\end{proposition}

Depending on the context, we denote by $C(f)$ either the deterministic communication complexity of $f$ or the randomized complexity for a fixed $\varepsilon=1/3$. 

Some of the canonical functions studied in communication complexity are the equality problem, denoted $EQ_n$ where the players output $1$ iff their inputs $x,y \in \{0,1\}^n$ are equal, the disjointness problem, denoted $DISJ_n$ where the goal is to check whether the $n$-bit strings interpreted as sets intersect or not, and the inner product problem $IP_n$ where the players need to output the inner product of their inputs modulo $2$.

The functions $DISJ_n$ and $IP_n$ are ``hard" functions for \CC, in the sense that almost all the input must be sent even when we allow for randomization, error and expected length. The following two bounds, which we will need later,
can be derived for example from \cite{YJKS04}, where the notion of information cost is used. 
\begin{theorem}\label{thm:av_std}
Any protocol for $DISJ_n$ or $IP_n$ with error $1/2 - \varepsilon$ has communication complexity
$\Omega (\varepsilon^2n)$.
\end{theorem}

%

\subsection{Streaming algorithms}
In the streaming model, the input comes as a stream to an algorithm whose task is to compute some function of the stream while using only a limited amount of memory and making a single or a few passes through the input stream.
See \cite{mutu} for a general introduction to the topic. If possible, the updates should also be fast. It was defined in the seminal work of \cite{ams} where the authors provided upper and lower bounds for computing some stream statistics. Since then, a plethora of results have appeared for computing statistics of the stream, as well as for graph theoretic problems. For the graph problems, we will assume that the graph is revealed to the algorithm as a stream, one edge at a time. In the more recent turnstile model, streams are made of both insertions and deletions,
and the goal is to compute some function that depends only the remaining elements (and eventually their respective frequencies). As we have said, any problem in the turnstile model can be solved via linear sketches~\cite{Li14}.

\subsection{The new model of Streaming Communication protocols}\label{sec:definition}\label{sec:model}
We show how to extend the original model of communication complexity to account for streaming inputs.
In the Streaming Communication $\SC$ model we consider that the inputs $x,y$ 
are not given all at once to the two players Alice and Bob but rather come as a stream. Moreover each player 
only has limited storage, $S$ bits of memory. In the randomized case,
the players also have access to a shared random bit string $r$ which may be infinite. They may use as many coins as they like from these strings.

A \emph{protocol} $\Pi$ in the streaming communication model is specified by four functions $\Phi^A,\Phi^B,\Psi^A, \Psi^B$. 
Each time slot $i$ is divided in two phases:
\begin{compactenum}
\item Each party receives a message from the other party ($m^B_i$ and $m^A_i$ resp.) and updates their memory (that was in state $\stt^{A}_i$ and $\stt^{B}_i$ resp.) according to the function $\Phi^{A}$ and $\Phi^{B}$ resp. This function also depends and the shared random string $r$, which is not restricted in size. 
\item Messages $m^A_{i+1}$ and $m^B_{i+1}$ are produced using the functions $\Psi^{A}$ and $\Psi^{B}$ resp., that depend on the current memory states $\stt^{A}_{i+1}$ and $\stt^{B}_{i+1}$ resp., the newly read input bit, and the randomness $r$. The messages might be empty and they could also be arbitrarily big in principle, though we will see in Proposition~\ref{proposition:size} that their size can be assumed to be $S+1$ without loss of generality.
\end{compactenum}

 The memory state $\stt^{A/B}_{i+1} \in \{ 0,1 \}^{S}$ and the next message $m^{A/B}_{i+1} \in \{0,1 \}^{*}$ can be defined recursively as
$\stt^{A/B}_{i+1} := \Phi^{A/B}(m^{B/A}_{i},\stt^{A/B}_i,r)$ and
$m^{A/B}_{i+1}:= \Psi^{A/B}(\stt^{A/B}_{i+1},x_{i+1}, r)$.
 Moreover, we assume that the streams end with a special EOF symbol and that once the streams are finished, the players only get one last round of communication, and then they have to output something.
%


\begin{definition}[SC protocols]
An SC protocol $\Pi$ uses $S$ bits of memory, $\rounds$ rounds, and $\bits$ bits when
\begin{compactitem}
\item The memory size of each player is at most $S$ bits;
\item The (expected) number of time slots where either $m_i^{A} \neq \emptyset$ or $m_i^{B} \neq \emptyset$ is at most $\rounds$;
\item The (expected) size of all exchanged messages is at most $\bits$ bits.
\end{compactitem}
The expectation is over the randomness of the protocol and worst-case over the inputs. An SC protocol is said to be {\em one-way}  if there is a single message from Alice to Bob after the streams have been received, and only Bob computes the function.
\end{definition}

Note, that our model carries an implicit notion of time due to the players reading their streams synchronously, and hence, the ability to send empty messages can be used to reduce communication \cite{Williams10}. However the gain is only logarithmic in the number of available time slots (see Section~\ref{sec:simclock}). We could have avoided such extra power, by defining a model where agents know when they should speak
or read a bit, based on the previous messages they received and their memory content.
Nevertheless, we opted for our model, as it is simpler to state and the necessary resources do not change by more than a factor logarithmic in the input size.

When we prove lower bounds or communication-memory tradeoffs, we do not consider the complexity of $\Phi^{A/B}$. These functions could be of arbitrarily high complexity.
To make things simple, we assume they are the same functions for every round $i \in [n]$, but they can depend on $n$. This  framework captures the streaming model as a special case, when the output depends on the stream of Alice only. 

\subsection{Properties of the SC model}\label{sec:generalities}

%

\label{sec:compr}

Several times in our proofs,
 we will consider an SC protocol and use it to solve problems in the standard models of CC.
 It is convenient to have a bound on how big the messages $m^{A/B}$ can be.
 
 The length of the messages $m^{A/B}$ could be very big in the SC protocol, but we now show that the SC protocol can be simulated replacing them by length $S+1$ messages.
 
 \begin{proposition}\label{proposition:size}
 In the SC model, we may always assume that the size of the messages is at most $S+1$ bits, up to redefining the transition functions $\Phi^{A/B}$.
 \end{proposition}
 \begin{proof}
 Consider a protocol $\Pi$ with associated functions $\Phi^{A/B}, \Psi^{A/B}$. 
%
The players can exchange their $S$ size memory and the last input symbol instead of the actual messages. 
Hence, it is possible to redefine functions $\Phi^{A/B}, \Psi^{A/B}$ and
 directly assume messages have length $\leq S+1$. 
The new equations with $S+1$ bit messages would read
$\stt^{A/B}_{i+1} := \Phi^{A/B}(\Psi^{B/A}(\stt^{B/A}_{i},y_{i}, r),\stt^{A/B}_i,r)$ and
$m^{A/B}_{i+1}:= \Psi^{A/B}(\stt^{A/B}_{i+1},x_{i+1},r)$.
\end{proof}

\label{sec:simclock}\label{lowerb}
Any protocol in the SC model can be simulated with another protocol in the usual CC model with a small overhead. Note that due to the implicit time in the SC model, we cannot immediately conclude that the SC model is harder than the usual communication model. 
Nevertheless, this time issue induces only an extra logarithmic factor. 

\begin{proposition}\label{prop:clocksimth}
We can simulate any protocol $\Pi$ in the SC model with parameters $S,R,T$ with another protocol $\Pi'$ in the normal communication model such that its communication cost $C(\Pi')$ is bounded as 
$C(\Pi') \leq  \bits (1 + 2\log n)$ and
$C(\Pi') \leq  \rounds (S + 2\log n+1)$.
\end{proposition}

We now compare the SC model to streaming algorithms, that is when there is a single player and a single stream.
There are various ways to combine streams $x$ and $y$ in a single stream.
Since $x_i$ is presented to Alice at the same time as $y_i$ to Bob, in a single player model $x_i$ should be presented just before $y_i$ to the player. This explains why we consider the interleaved streaming model.
\begin{definition}
 Let $\intstr(f)$ be the amount
of memory required for a streaming algorithm to compute $f$ where the input stream is $x$ and $y$ interleaved, 
that is to say, $x_1,y_1,x_2,y_2,\ldots,x_n,y_n$.
\end{definition}

It turns out that interleaving streams instead of concatenating streams may affect the memory requirement of the function for a standard streaming algorithm by an exponential factor (see Appendix~\ref{sec:int_vs_std}).
\begin{theorem}\label{thm:example2}
There is a function $f$ such that $ \intstr(f) =  \Omega(n)$, whereas there is a streaming algorithm to compute $f$ with memory $O(\log n)$
when streams are concatenated.
\end{theorem}

First let us observe that $S^{\sd}(f)$ provides a lower bound in the SC model.
\begin{proposition}\label{summem}
Let $f$ be a function. Then, any protocol in the SC model for the function $f$, where Alice and Bob use memories of size $S$, must have
$ 2S \geq S^{\sd}(f).$
\end{proposition}

It is natural to ask if there is a polynomial relation bounding the parameters $S, R, T$ of a protocol in the SC model in terms of the streaming complexity $\intstr(f)$ and the communication complexity $C(f)$, 
at least when $S=O(\intstr(f))$.
This appears to not hold in general (See Appendix~\ref{sec:gap_sc}).

\begin{theorem}\label{thm:example}
There exists a function $f$ such that $\intstr(f)=C(f)= O(\log n)$ but any protocol computing $f$ in the $SC$ model must have $R \cdot S = \Omega(n)$. This holds for $S=\Omega(\log n)$.
\end{theorem}

\section{Communication primitives}\label{sec:comm}
We provide a general theorem that provides tight tradeoffs both in the deterministic and randomized case for a variety of functions, including $DISJ_n, IP_n, TRIBES_n$. 


\subsection{A general lower bound}\label{sec:genlb}
In this section we show a general result that gives a lower bound for a large class of functions. We will obtain the lower bounds for the usual primitives $DISJ, IP, TRIBES$ as a corollary.
We treat $\varepsilon$ as a constant.
Let us start with a high level description of the result. 

Assume $f$ can be written as the composition of an outer function $G$ and inner gadgets $g_\ell$. We want to 
lower bound the communication needed to compute $f$ in the SC model in terms of the communication complexity of 
each $g_\ell$ and the available memory. 

\begin{definition}
We call a function $G$ on $L$ variables \emph{non trivial} if the following holds.
There exists a word $a \in \{0,1\}^{L}$ such that for all $\ell$ there exists a postfix 
$b_{\ell}\in \{0,1\}^{L-\ell-1}$ such that $G(a_{\leq \ell}ub_{\ell})$
 depends on the bit $u$. 
More formally
$ G(a_{\leq \ell}0b_{\ell}) \neq G(a_{\leq \ell}1b_{\ell})$.
\end{definition}

This may look as a restrictive condition. 
 In fact, most natural functions that depend on every bit are "non trivial" in this sense. 
 For the function $\bigoplus$, $a,b$ 
 can be chosen arbitrarily. 
 The functions $\OR$ and $\AND$ are also non trivial. For instance for $\AND$, $a=b=1^L$ will do.

We borrow the next definition from \cite{YJKS04} (extending it slightly).

\begin{definition}[Block-decomposable functions.]
 Let $I_1,\ldots,I_L$ be an interval partition of $[n]$, which we refer to as \emph{blocks}. For $\ell \in [L]$, 
 let $t_\ell=|I_\ell|$ be the length of $I_\ell$.
 Given strings $x,y \in \{0,1\}^n$, write $x^\ell$ (resp. $y^\ell$) for the restriction of $x$ (resp. $y$) to indices in block $I_\ell$.
We say $f : \{0,1 \}^n \times \{0,1\}^n\rightarrow \{0,1\}$ is \emph{$G$-decomposable} with primitives $(g_\ell)$, where $G :
\{0,1\}^{L} \rightarrow \{0,1\}$ and $g_\ell : \{0,1\}^{t_\ell} \times \{0,1\}^{t_\ell}  \rightarrow \{0,1\}$,
if for all inputs $x,y$ we have
$ f(x,y)=G(g_1(x^1,y^1),g_2(x^2,y^2),\ldots,g_L(x^L,y^L))$.
\end{definition} 
 
\begin{theorem}\label{thm:genlem}

 Assume the function $f:\{0,1 \}^n \times \{0,1\}^n\rightarrow \{0,1\}$
 is $G$-decomposable with primitives $(g_\ell)_{\ell \in L}$, and that $G$ is non-trivial.
Let $C_{\varepsilon+\delta}(g_\ell)$ be the worst-case randomized communication complexity of $g_\ell$ in the usual communication model. Then,
any randomized protocol computing $f$ with error $\varepsilon$ in the SC model with $S$ bits of memory, $\rounds$ expected communication rounds and $\bits$  expected bits of total communication, must have
$$\rounds \geq  \sum_{\ell  \leq L} \frac{\delta C_{\varepsilon+\delta}(g_\ell) - \co S}{S + 2\log t_\ell+1}, \quad
\bits  \geq  \sum_{\ell \leq  L} \frac{\delta C_{\varepsilon+\delta}(g_\ell) - \co S}{1 + 2\log t_\ell}.$$
 We can get a similar bound in the deterministic case, where we use the deterministic communication complexity of the $g_\ell$'s. Last, if $G$ is $\bigoplus$ we may remove $\delta$ from the above bounds, changing the complexities 
 $C_{\epsilon+\delta}(g_\ell)$ to  $C^{avg}_\varepsilon(g_\ell)$.
\end{theorem}

\subsection{Applications}
Before proving Theorem \ref{thm:genlem}, we give a few corollaries. Note that the upper bounds are trivial.
Remind the function $TRIBES$, which is an $AND$ of Set Intersections is defined by
$  TRIBES_n(x,y) := \AND_{i \leq \sqrt{n}}\circ \OR_{j \leq \sqrt{n}} \left(x_{ij} \bigwedge y_{ij}\right)$.
%
\begin{theorem}\label{thm:corgenlem}
  Any randomized protocol in the SC model that computes the function $IP_n$, $DISJ_n$, or $TRIBES_n$ and uses $S$ memory and $R$ communication rounds must have 
$ R \cdot S =\Omega(n)$.
 \end{theorem}
 \begin{proof}
We start by the argument for $DISJ_n$.
We write 
$DISJ_n(x,y) = AND_{\ell=1}^{n/k}\left(DISJ_{\ell}(x^{\ell},y^{\ell})\right)$,
and use the previous result with $G=AND$ and $g_{\ell} = DISJ_{k}$. It follows from Theorem \ref{thm:av_std} that
$ R_{\varepsilon+\delta}(DISJ_k) = \Omega(k)$. We omit the dependency on $\varepsilon$ and $\delta$ in the term $\Omega()$, treating these parameters as fixed constants. The number of rounds for any randomized protocol in the SC model is at least
$ \sum_{\ell = 1}^{\frac{n}{k}} \frac{\Omega( k) - \co S}{S + 2\log k}$.
We get the result choosing $k = \Omega(S)$.

In the case of $IP_n$, the function $f = IP_n$ is the composition of $G = \bigoplus$ over $\frac{n}{k}$ coordinates with $g_{\ell}=IP$ (for each $\ell \leq \frac{n}{k}$), over $k$ coordinates.
 Theorem \ref{thm:av_std} gives $\ravg_{\varepsilon}(g_\ell) \geq \Omega(k)$. 
 Theorem \ref{thm:genlem} yields the bound, taking $k=10S$.
We omit the case of $TRIBES_n$ as it follows from a similar argument.
 \end{proof}

\section{Approximate~Matching in the Streaming Communication model}\label{sec:matching}

%
%

\label{sec:approxm}

The main problem we are going to consider is that of computing an approximate matching. The stream corresponds to edges (in an arbitrary order)
of a bipartite graph $G=(P,Q,E)$ over vertex set $P,Q$, and the algorithm has to output a collection of edges which forms a matching. All edges in the output have to be in the original graph. In the vertex arrival setting, each vertex from $Q$ arrives together with all the edges it belongs to.
Our goal is to understand what is the best approximation ratio we can hope for, for a given memory (and message size). 

%

We start by some notations.
In a graph $G=(P,Q,E)$, if $U \subseteq P \cup Q$ and $V \subseteq P \cup Q$ are subsets of the vertices, we denote by $E(U,V) \subseteq E$ the edges with endpoints in $U$ and $V$. We also denote by $OPT(G)$ the maximum size of a matching in $G$.

Observe now that when communication can occur at any step, the greedy algorithm, which is currently the best algorithm in the standard streaming model, can be implemented easily by having Alice communicate to Bob every time she adds an edge to her matching. 
\begin{proposition}
The greedy algorithm, which achieves a $2$-approximation, can be implemented using $n \log n$ bits of communication and $n$ rounds.
\end{proposition}
Thus, we now focus on the one-way SC model, where the communication is restricted to happen once the streams have been fully read. In this setting, we will get different lower and upper bounds than in the streaming model.
We start by a positive result.

\begin{theorem}[Greedy matchings]\label{thm:greedy}
If Alice sends Bob a maximal matching of her graph, then Bob can compute a $3$-approximation. In particular
a $3$-approximation can be computed in a deterministic one-way SC protocol using $\bigO(n\log n)$ memory and message size.
\end{theorem}
\begin{proof}
Let $G_1,G_2$ be the respective graphs that Alice and Bob get, and
let $M_1,M_2$ be their respective computed maximal matchings.
We show that there is a matching in $M_1\cup M_2$ of size at least $|OPT(G)|/3$.

The proof goes in two steps. Let $\ell$ be the size of $V(M_1 \cup M_2)$.
We prove first that there is a matching of size $\geq \frac{\ell}{3}$ in the graph $M_1 \cup M_2$,
and then that the number of edges in $OPT(G)$ is at most $\ell$,

The first part is easy. Observe that $M_1 \cup M_2$ has maximal degree $2$. Then there must be a matching of size at least $\ell/3$ from Theorem $7$ in \cite{Biedl}.

For the second part, we construct an injection $OPT(G) \hookrightarrow V(M_1 \cup M_2)$. Let $e=(u_1,u_2)\in OPT(G)$. Assume without loss of generality that $e\in G_1$. 
Then either $u_1$ or $u_2$ is matched in $M_1$ (maybe both), by maximality of $M_1$. Map $e$ to (one of) its matched endpoint in $M_1$.
\end{proof}

Our lower bound is obtained using a black box reduction that we develop in the following sections.
It is a direct consequence of the combinaison of Theorem \ref{thm:reduction} and Theorem \ref{thm:kapramain}. 



\begin{corollary}[A $(e+1)/(e-1)$ lower bound]\label{cor:lowerb}Any protocol achieving a ratio of $\frac{e+1}{e-1} - \eta \approx 2.16- \eta$, for some constant $\eta$, in the vertex arrival setting needs communication $n^{1+ \Omega(1/\log \log n)}$
where the hidden constant in the $\Omega(.)$ depends on $\eta$.
\end{corollary}

\subsection{Hard distributions for streaming algorithms}

Our notion of hard distribution is tailored to capture the distributions appearing in \cite{Goel12, Kapralov13}.
They are distributions over streams of graphs, that is over graphs and edge orderings.

We will use the following definition for constructing families of hard distributions when $n \rightarrow \infty$ and $\alpha,\eta$ are fixed. Therefore $O()$ and $o()$ notations have to be understood in that context.


\begin{definition}[Hard distribution]\label{def:natural}
A distribution  $\mu$ over streams of bipartite graphs $G = (P,Q,E)$
is an \emph{$(\alpha, n, m(n), \eta)$-hard distribution}
when  $P$ and $Q$  are sets of size $n$ 
and the following holds 
\begin{compactenum}
\item There is a cut $\cutl, \cutr$ of vertices 
such that
$\lvert \cutl \cap Q\rvert + \lvert \cutr \cap P \rvert \leq (1- \alpha + \eta)n$.

\item There is a matching $M$ of size $(1-\eta)n$ in $G$ that can be decomposed into $M_0 \cup M'$
such that \\
(i) $P_0:=V(M_0) \cap P$ and $Q_0:=V(M_0) \cap Q$ are of fixed size (in the support of $\mu$)
larger than $(\alpha -\eta)n$ and smaller than $\alpha n$; and \\
(ii) $P_0 \subseteq \cutl$ and $Q_0 \subseteq \cutr$. 
\item Every streaming algorithm $\algo$ with $o(m)$ bits of memory that outputs $E^*$ with $E^{*}\subseteq E$
must satisfy
$\lvert E^* \cap  E((\cutl \cap P) \times (\cutr \cap Q)) \rvert = o(n)$ with probability $1-o(1)$.
\end{compactenum}
\end{definition}
In particular, a streaming algorithm with small memory can only maintain on hard distributions
a small fraction of edges $E((\cutl\cap P) \times (\cutr\cap Q))$ and therefore of $M_0$.
In addition, observe that for every matching $E^*$ and cut $(\cutl,\cutr)$ (see Figure \ref{fig:maxflow})
$\lvert M \rvert  \leq \lvert \cutl \cap Q \rvert + \lvert \cutr\cap P\rvert + E((\cutl\cap P) \times (\cutr\cap Q))$.
Thus edges from $E(\cutl\cap P \times \cutr\cap Q) $ are also crucial for obtaining a good matching. 

The existence of hard distributions is ensured by~\cite{Goel12, Kapralov13}.
The hard distributions families are not exactly presented as we present them here. The following Theorem follows from
results in  \cite{Kapralov13} involving more parameters,
for the purpose of the construction itself. We disregard those since we use the existence of hard distribution families as a black box.

\begin{theorem}[\cite{Kapralov13}]\label{thm:kapramain}
For all $\eta>0$ and  $n$,
there is a $(1/e, n,  m(n),\eta)$-hard distribution $\mu$ over graphs of $n$ vertices with  $m(n) = n^{1+\Omega(\frac{1}{\log \log n})}$, where the notation $\Omega(.)$ hides a dependency on $\eta$. 
\end{theorem}

\begin{proof}[Sketch of proof]
Specifically, point $(1)$ of our definition follows from \cite[Lemma $13$]{Kapralov13},  point $(2)$ follows from \cite[Claim $12$]{Kapralov13},
and point $(3)$ follows from \cite[Lemma $14$]{Kapralov13}. 
\end{proof}

\subsection{Lifting hard distributions  to streaming communication protocols}



Let $\mu$ be an $(\alpha,n, m,\eta)$-hard distribution for streaming algorithms. 
We will show how to extend it to a distribution $\mu_2$ for the two party version of the approximate matching problem. At a high level, we give the players two copies of the same graph $G$ randomly chosen according to $\mu$,
but embedded into two different but overlapping subsets of vertices (see Figure \ref{fig:double2}). 
The non-overlapping parts correspond to edges from $E(\cutl\cap P \times \cutr\cap Q)$ (see Definition~\ref{def:natural}),
and are therefore hard to maintain, but necessary to setup a large matching.


From now on, identify $P$  with the set $[n]$.
Set $\beta := 1+\alpha$.  
Our labelings are defined over vertex set $P' \times Q' := [\beta n] \times [\beta n]$,
and are encoded by injections from $[n]$ to $[\beta n]$
 (where for simplicity $\beta n$ is understood as an integer).
Given a hard distribution $\mu$, we define a distribution $\mu_2$ as follows.
\begin{definition}[The distribution $\mu_2$]\label{def:mu2}
Let $\mu$ be a hard distribution, where $P$ and $Q$ are identified with $[n]$.
Then sampling  a bipartite graph over vertex set $P'\times Q'= [\beta n] \times [\beta n]$ from $\mu_2$ is defined as follows
\begin{compactitem}
\item Sample $G \sim \mu$. Let $(\cutl,\cutr)$ and $P_0,Q_0$ be the corresponding cut and sets from Definition~\ref{def:natural}.
\item Sample $\sigma$, $\tau$ uniformly at random such that
 $\sigma(P_0)\cup\tau(P_0)=\emptyset=\sigma(Q_0)\cup\tau(Q_0)$ are disjoint, and  $\sigma,\tau$ are equal on $P \setminus P_0$.  Such injections $\sigma,\tau$ are called \emph{$G$-compatible}.
 \end{compactitem} 
In addition, define $G_\sigma := (\sigma(P), \sigma(Q),E_\sigma)$, where 
$E_{\sigma} = \{ \left( \sigma(u), \sigma(v)\right) \mid (u,v) \in E\}$, and $G_\tau$ similarly.
Alice is given $G_{\sigma}$ and Bob is given $G_{\tau}$ with the same order as under $\mu$.
\end{definition}

In this construction, observe that edges sent to Alice and Bob may overlap.
In fact the distribution can be tweaked to make edges disjoint using a simple gadget, while preserving the same lower bound.
We present this in Appendix~\ref{sec:disjoint}. We can now state our main result for the reduction.

\begin{theorem}[Generic reduction]\label{thm:reduction}
If there exists an $(\alpha, n, m, \eta)$-hard distribution for approximate matching, then any protocol in the one-way SC model whose approximation ratio is
$\frac{1-\alpha}{1+\alpha} - O(\eta)$
has to use $\Omega(m)$ memory.
\end{theorem}
\begin{proof} 
Define a cut for the two player instance as 
$\cutl' := \sigma(\cutl)\cup \tau(P_0)$,
$\cutr' := \sigma(\cutr) \cup \tau(Q_0)$.
 It follows from the max-flow min-cut argument that any matching $E^{*}$ has size at most
$\lvert E^* \rvert \leq \lvert \cutr' \cap Q' \rvert + \lvert \cutl' \cap P'\rvert +
 \lvert E^* \cap E(\cutl'\cap P' \times \cutr'\cap Q') \rvert$.

Moreover note that by construction $\lvert \cutr' \cap Q' \rvert =  \lvert \sigma(\cutr )\cap Q' \rvert$.
Indeed $\tau(P_0) \cap Q' = \emptyset$.
Then we can write $
 \lvert \sigma(\cutr )\cap Q' \rvert = \lvert W \cap Q\rvert$ and similarly for $ \lvert \cutl' \cap P'\rvert $
(see also Figure \ref{fig:double2}). It follows that
$$ \lvert \cutr '\cap Q' \rvert + \lvert \cutr' \cap P'\rvert
=  \lvert \cutr \cap Q \rvert + \lvert \cutr \cap P\rvert 
 \leq \left(1-\alpha+\eta\right)n.
$$

The set of edges the protocol outputs is
included in $E^*_A \cup E^*_B$ by definition. 
Under the high probability event that these
sets only have an overlap of $o(n)$ with
the ``important edges"
$E(\cutl'\cap P \times \cutr'\cap Q)$
(see Lemma \ref{lem:union} below) then, if $E^*$ is the output matching by a protocol using $o(m)$ memory, then using Lemma \ref{lem:union} the matching $E^* \subseteq E^*_A \cup E^*_B$ is of size at most $\left(1-\alpha + \eta\right)n +o(n) \leq \left(1-\alpha +2\eta\right)n$ (for large enough $n$).

On the other hand, there is a matching between $\sigma(P)$ and $\sigma(Q)$ of size $(1- \eta)n$ and a matching of size $(\alpha - \eta) n$ between $\tau(P_0)$ and $\tau(Q_0)$
(using Point $(2)$ in the definition of a hard distribution for approximate matching, Definition \ref{def:natural}). 
We identified $\sigma(P) \sqcup \tau(P_0)$ with $[\beta n]$ and hence under  $\mu_2$ there is a matching of size $(\alpha-\eta)n + (1-\eta)n=\beta n-2\eta n$. This shows that the approximation ratio of a protocol
using $o(m)$ memory is at most
$\frac{\beta n -2 \eta n}{\left(1-\alpha +2\eta\right)} =\frac{1+\alpha}{1-\alpha} -\bigO(\eta)$.
\end{proof}

\begin{lemma}\label{lem:union}
Let $\algo$ be a protocol for the two party case, i.e. a pair of algorithms for Alice and Bob $\algo = (\algo_A, \algo_B)$. Let $E^{*}_A$ (resp. $E^{*}_B$) denote the edges $\algo_A$ (resp. $\algo_B$) outputs, assuming  only $o(m)$ memory is used.
With probability $1-o(1)$ over the choice of  $\sigma, \tau, G$, or alternatively with probability $1-o(1)$ under $\mu_2$, it holds that
$$\lvert E^*_A \cap E(\cutl'\cap P' \times \cutr'\cap Q')  \rvert  = o(n),
\quad
\text{and similarly}
\quad
\lvert E^*_B \cap E(\cutl'\cap P' \times \cutr'\cap Q')  \rvert = o(n).$$
\end{lemma}
\begin{proof}
The proof consists in building from $\algo_A$, and similarly $\algo_B$, a streaming algorithm for $\mu$.
Indeed, for inputs over vertices $[n]$ distributed according to $\mu$, simply pick a random $\sigma$
apply it to the input, 
 and run $\algo_A$ on the graph $G_{\sigma}$. Then output $E_{0}^* := \sigma^{-1}(E^*_{A})$.  

First observe that the distribution of $\sigma$ and the distribution of $G$ are independent. Therefore, conditioned on $G$, the injection $\sigma$ is uniform.
It follows 
that $\mu_2$'s first marginal is also the distribution of $G_\sigma$, where $G \sim \mu$ and $\sigma$ is uniform and independent.

 
Then, using Definition~\ref{def:natural}, with  probability $1-o(1)$ over the choice of $\sigma$ and $G \sim \mu$, we obtain that 
$\lvert E^*_A \cap  E(\cutl'\cap P' \times \cutr'\cap Q')\rvert =\lvert E^*_0 \cap  E(\cutl\cap P \times \cutr\cap Q)\rvert =o(n)$,
which concludes the proof.
\end{proof}

\begin{figure}[!ht]
    \centering
    \includegraphics[width=0.3\textwidth]{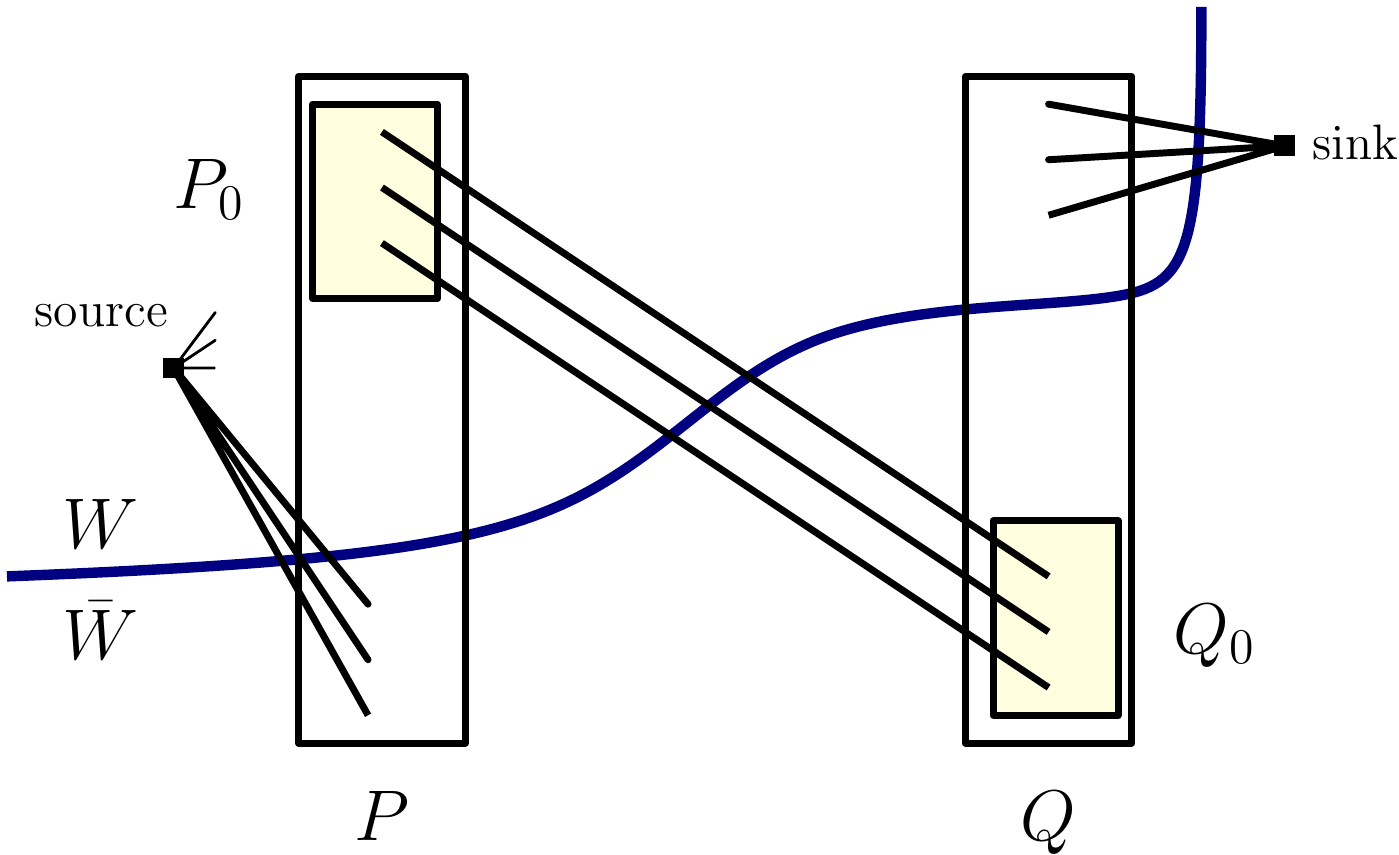}
        \caption{The figure shows how the maxflow mincut theorem is used to argue about the size of a matching in a bipartite graph $G$ over vertex set $P\times Q$ drawn from a distribution $\mu$. Only edges from the cut are drawn. The source and sink are added for the sake of the argument, they are not part of  $G$.} \label{fig:septree} \label{fig:maxflow}
\end{figure}

\begin{figure}[!ht]
    \centering
    \includegraphics[width=0.3\textwidth]{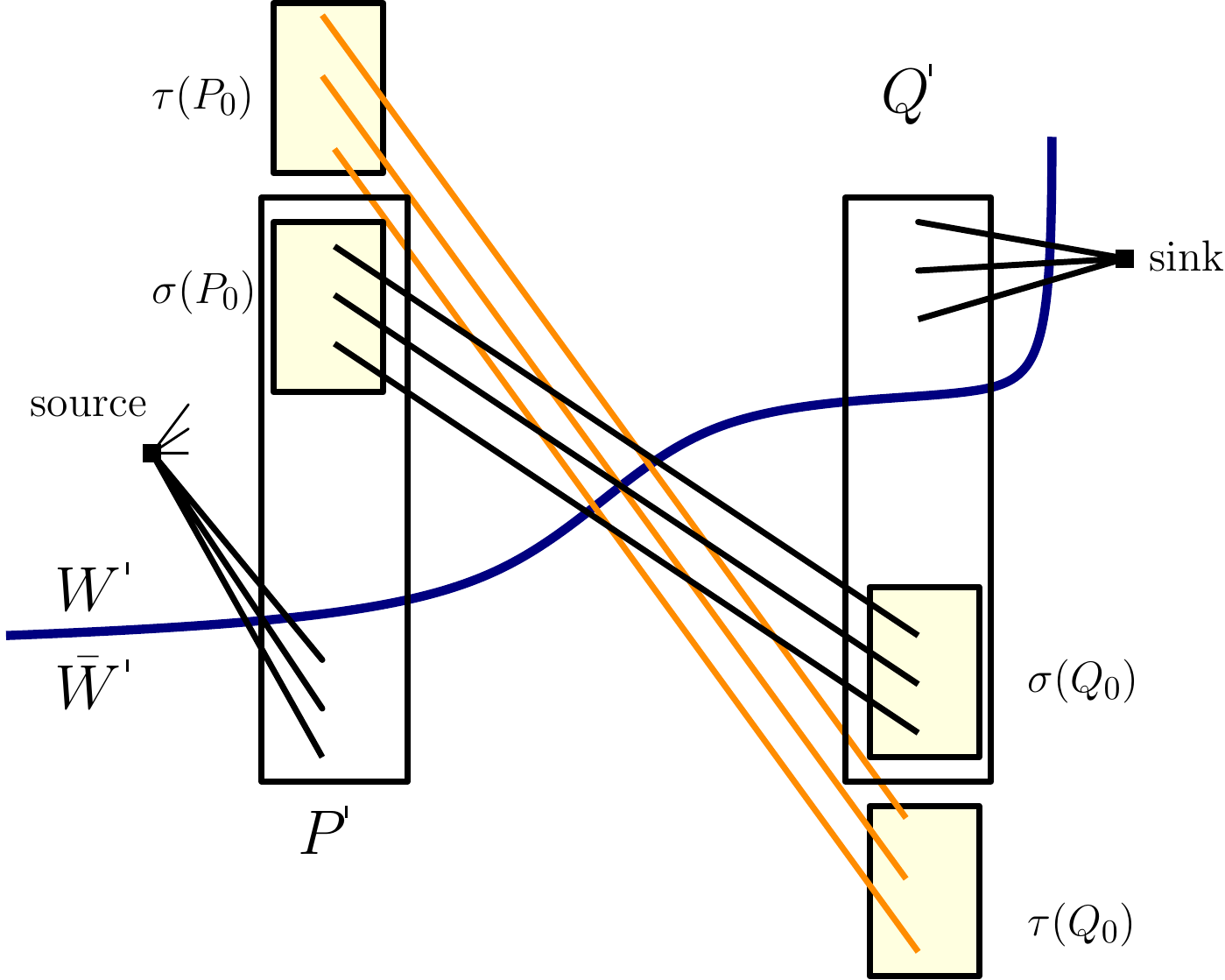}
        \caption{The construction of $\mu_2$.} \label{fig:septree} \label{fig:double2}
\end{figure}

\bibliographystyle{plain}
\bibliography{refs}

\appendix


\section{Missing proofs from Section~\ref{sec:generalities}}

\subsection{Proof of Proposition~\ref{prop:clocksimth} and Proposition~\ref{summem}}
\begin{proof}[Proof of Proposition~\ref{prop:clocksimth}]
The following protocol $\Pi'$ is a simple simulation of the protocol $\Pi$ in the usual communication model. 
We assume that both players know when the protocol has ended.
Both players update a variable named $t$ supposed to emulate time. Initially,  $t$ is set to $t:=0$.
\begin{compactenum}
 \item If Alice and Bob know that the protocol $\Pi$ has ended, then they output as in $\Pi$ and exit.
 \item Otherwise, Alice and Bob both declare at which time after $t$, say $t^A,t^B$ they would send a nonempty message in $\Pi$ given
 the current transcript and time $t$.
 \item Whoever has the minimum of $(t^A,t^B)$, sends the message in $\Pi$ corresponding to this time. 
 If $t^A=t^B$, then both send their message.
 \item They update $t= \min (t^A, t^B)$ and go to $1$.
\end{compactenum}
The protocol $\Pi'$ outputs exactly as in $\Pi$ and we have $C(\Pi') \leq R \times 2\log n + B$.
The factor  $2$ in the expression $2\log n$ comes from both Alice and Bob sending a time at step $(2)$ in $\Pi'$.
We conclude by noticing that $R \leq T$ and $T \leq R \cdot (S+1)$ (by Proposition \ref{proposition:size}).
\end{proof}

\begin{proof}[Proof of Proposition \ref{summem}]
The streaming problem associated to $f$ can be solved, using an
$SC$ protocol and only the memory of the players needs to be written down.
The input to this problem is $x$ and $y$ interleaved, 
by definition of the SC
model. 
The algorithm does not need to store the messages $m^{A/B}_i$. It simply computes them on the fly
and does the updates of $\stt^A,\stt^B$. Note $2S$ bits suffice to write $
(\stt^A,\stt^B)$.
\end{proof}

\subsection{A function that is easy in Streaming and CC but not for Streaming Communication protocols - proof of Theorem \ref{thm:example2}}\label{sec:int_vs_std}


%

Rather than exhibiting a function $f$ with the desired properties, in the proof, we consider languages. This is without loss of generality as we can freely move from languages to functions by considering the indicator of the language (for a fixed size of input, $n$).

Let DYCK($2$)
be the language consisting of all well parenthesized words with two different types of parenthesis.
We will restrict our attention to $L:=  DYCK(2)_3\subseteq DYCK(2)$, which consists of words with 
only $3$ alternations. Namely, $x \in L$ iff $x$ can be written in the form $x=u_1v_1u_2v_2u_3v_3$ and the $u_i$'s
are only made of opening
parenthesis whereas the $v_i$'s are only made of closing parenthesis.

Using linear hashing, an algorithm for $L$ is known in the standard streaming model. It 
only needs $\bigO(\log n)$ memory (see \cite{MagMN}) so that $\str(L) = \bigO(\log n)$. However,
if the word $x$ is delivered in an interleaved fashion, much more memory is needed.

We now show that $\intstr(L) =  \Omega(n)$.
 The proof follows by providing a direct reduction to a generalization of the INDEX problem from communication complexity.
 \begin{definition}
 In the INDEX problem, Alice is given a string $u \in \{0,1\}^n$, and Bob is given $u_1,\ldots,u_k,b\in\{0,1\}^{k+1}$
 for some $k$ which varies. The goal is for Bob to output $b=x_{k+1}$
 \end{definition}
 It is known that any one-way protocol where only Alice is allowed to speak, requires at least $n$ bits
 \eq{ 
 R^{one-way}(INDEX) = n.
 } 
 We use the notation $\overline{u}$ to denote the same word as $u$ but with all open parenthesis replaced
 by closing ones and viceversa.
 
 Let $n_1 \in \NN$ and set $n := 3n_1+2$.
 We will restrict inputs to the following set
 \eq{
 A_n = \{ x \in \{0,1\}^n, x=uvb\overline{b}\overline{v}w\overline{w}\overline{u}w'\overline{w'} \mid  b \in \{0, 1\}, u \in \{0,1 \}^{n_1}, v=u_1,\ldots,u_k, k\leq[n_1]\}
 }
\begin{figure}[h]\label{fig_dyck}
\centering{
\includegraphics[scale=0.5]{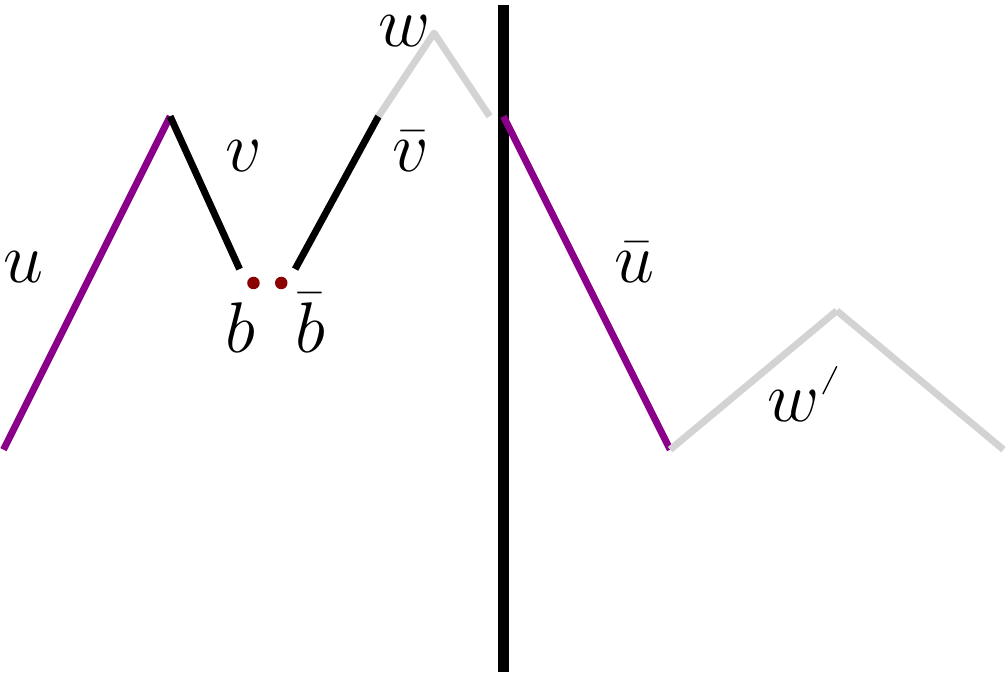}}
\caption{A word from the class $A_n$. It is represented according to the conventions of \cite{MagMN}.}
\end{figure}

 The words $w,w'$ are fixed and made of zeros only. This corresponds to a padding. By
 setting the length of $w, w'$ such that $|u| + 2(|v| + 1 + |w|) = |u| + 2|w|' = 3n_1+2$,
 we enforce that all words in $A_n$ have the same length $n:= 3n_1+2$.
Also, note that the middle of the word $x=uvb\overline{b}\overline{v}w\overline{w}\overline{u}w'\overline{w'}$ is exactly after $\overline{w}$.
 
 We easily see that $x=uvb\overline{b}\overline{v}w\overline{w}\overline{u}w'\overline{w'} \in A$ is in $L$ iff $b=x_{k+1}$.
 In the interleaved problem, the first part of the stream corresponds to $(u,\overline{u})$. The second part
 corresponds to  $(vb\overline{b}\overline{v}w\overline{w}, w'\overline{w'})$.
 The core lemma is
 \begin{lemma}
 Assume there is an interleaved streaming algorithm that recognizes $A$, using
 $m$ bits. Then, there is a one-way protocol to solve INDEX$_{n_1}$ with $m$ bits,
 hence $m \geq n_1$.
 \end{lemma}
 \begin{proof}
 We show how to derive such a protocol for INDEX.
 Let $u \in \{0,1\}^{n_1}$ be Alice input and $u_1,\ldots,u_k,b$ be Bob's input.
 Let $v= u_1,\ldots,u_k$.
 Alice pretends she is getting $(u,\overline{u})$ in the interleaved streaming problem.
She then sends the content of her memory to Bob, which pursues as if he was getting 
$(vb \overline{b}\overline{v} w\overline{w}, w'\overline{w'})$ and he answers as the streaming algorithm would.
The correctness follows from the remark that $x=uvb\overline{b}\overline{v}w\overline{w}\overline{u}w'{w}' \in A_n$ iff
$b=x_{k+1}$.
 \end{proof}
 Notice that, on the class of inputs $A_n$ we considered $n_1\sim \frac{n}{3}$.
 It follows from the lemma, that any algorithm for $DYCK(2)_3$ must use $\Omega(n)$ bits of memory. 


\subsection{A  function that is easy for streaming protocols but not in the interleaved variant - proof of Theorem \ref{thm:example}}\label{sec:gap_sc}
%

First note that the memory $S$ available to an $SC$ protocol should be $\geq \log n$, otherwise we know from \ref{lowerb} the problem is not solvable.
Let $X,Y',Y,X' \in \{0,1\}^{\frac{n}{2}}$. Let $f(XY',YX')$ be the function 
\eq{ 
f(XY',YX') := IP(X,Y) \bigwedge X=X' \bigwedge Y=Y'.
}
\begin{proposition}
 The randomized communication complexity of $f$ is $\bigO(\log n)$.
\end{proposition}
\begin{proof}
A communication protocol to solve $f$, has Alice and Bob check the two equalities $X=X'$ and $Y=Y'$ using $\bigO(\log n)$ bits.
The inner product can then be evaluated without communicating, since both players know $X$ and $Y$.
\end{proof}
\begin{proposition}
 In the interleaved streaming model, which corresponds to the two streams combined coming as one, the problem can be solved using $\bigO( \log n)$ bits of memory
 \eq{
 \intstr(f) = \bigO(\log n).
 }
\end{proposition}
\begin{proof}
A streaming protocol checks $IP(X,Y)$ bit by bit, on the fly as they arrive. It also produces a linear
hash of length $O(\log n)$ for $X$ and $Y$ so as to check $X=X'$ and $Y=Y'$.
\end{proof}
\begin{proposition}\label{proposition:lastex}
Any protocol in the $SC$ model that computes $f$ using $S$ bits of memory and $R$ communication rounds must have
 \eq{
 R \cdot S= \Omega(n).
 } 
\end{proposition}
The proof is very similar to the one of Theorem \ref{thm:genlem} (see Appendix B), so we omit some details and advise the reader to read that proof first. 

Inputs will be partitioned into blocks of equal size $\alpha S= \Omega(S)$.  The parameter $\alpha$ is a constant to be determined later.
 Previously
we refered to the 
$\ell$th block as indices in $[\ell \alpha S,(\ell+1) \alpha S]$. Now we call the $\ell$th block indices
in $ [\ell \alpha S,(\ell+1) \alpha S] \cup [\frac{n}{2} + \ell \alpha S, \frac{n}{2} + (\ell+1) \alpha S]$,
for reasons that will soon be clear. The parameter $n$ stands for the size of the input,
which corresponds to $|X| + |Y|$.

 The construction is by induction.
 More precisely, we argue that there must exist inputs $x,y$ such that some communication happens
 on each of the $\ell$ blocks \emph{on average over the randomness}, otherwise we obtain a contradiction 
 by deriving a protocol for $IP_{\alpha S}$ using less than $\alpha S$ bits.
 
Let us fix $\Pi$ a protocol for $f$ in the $SC$ model. 
 \begin{lemma}\label{lem:mainex}
For any given $x^{\leq \ell},y^{\leq \ell}$, the inputs on block $\ell+1$ denoted
$x^{\ell+1}, y^{\ell+1}$ can be chosen such that 
when running $\Pi$ on $(x^{\leq \ell+1}wy^{\leq \ell+1}w',y^{\leq \ell+1}w'x^{\leq \ell+1}w)$, 
  where $w$ is some arbitrary word,
  the expected number of rounds on block $\ell+1$ is at least $\Omega(1)$, if $\alpha$ is chosen big enough.
 \end{lemma}
 
 The proof follows from the lemma by using it repeatedly to build inputs such that the lower bound holds simultaneously for every block.

\begin{proof}[Proof of Lemma \ref{lem:mainex}]
The proof of Proposition \ref{proposition:lastex} proceeds by reducing to $IP$
in the standard model of $CC$.

The following two paragraphs assume the reader is familiar with the Section \ref{sec:genlem}. When adapting the strategy we used in Section \ref{sec:genlem}, the protocol needs to be changed after step \ref{p3},
it is not clear that Bob can finish the simulation. 
Computing $f$ using $\Pi$, after step \ref{p3}, Bob needs to plug in values for $X',Y'$. 
By definition of the function $f$, the protocol will output $IP(X,Y) = IP(x^{\ell +1},y^{\ell+1})$ only if Bob is 
also able to pretend he is getting $X'=X$ and $Y'=Y$ or in other words if he can plug $x'^{\ell+1} = x^{\ell+1}$
and $y'^{\ell+1} = y^{\ell+1}$ as SC inputs for $f$.

However, he only knows his input $y^{\ell+1}$ and not $x^{\ell+1}$. We get around this issue by the following trick. After step \ref{p3},
Bob also sends Daniel's state to Alice so that \emph{she} does the updates of his memory pretending he is getting $X'=X$ (she knows $X$).
In parallel Bob updates Carole's memory pretending she is seeing $Y'$.

Let us now give more details.
Consider $\Pi$ over inputs starting with $x^{\leq \ell+1},y^{\leq \ell+1}$. As before, we denote by $R_{\ell+1}$
 the expected number of rounds, on block $\ell +1$.

%

We are going to find a protocol 
 $\Pi_{\ell+1}$ 
 for $IP_{\alpha S}$, in the usual model (without streams and bounded memory)
 using less than $R_{\ell+1} \times (S+log \alpha S) + 20S$ bits in expectation and error $\leq \varepsilon$. 
In turn, using Theorem \ref{thm:av_std} this quantity is greater than $c\alpha S$ if $\epsilon$ is small enough and $c$ is a small enough constant. Choosing a fixed $\alpha$ bigger than $100c^{-1}$ leads to $R_{\ell+1}\geq 80 \frac{S}{S + \log \alpha S} = \Omega(1)$ and this concludes the proof.
 
It only remains to explain how protocol $\Pi_{\ell+1}$ is derived from $\Pi$, and this is what we do next.
 
 \RestyleAlgo{boxruled}
\LinesNumbered
 \begin{algorithm}[ht]
 \caption{The protocol $\Pi_{\ell+1}$ solving $IP_{\alpha S}$, based on $\Pi$.}
\label{r1}
  Over inputs $x^{\ell +1},y^{\ell+1}$, Alice simulates the SC protocol $\Pi$
  on $x^{\leq \ell},y^{\leq \ell}$ for players Carole and Daniel. 
  \\
  \label{r2} Then Alice sends Bob the current 
  memory state of Daniel and message
 of the round from Carole to Daniel.
  \\
  \label{r3} Alice and Bob jointly simulate $\Pi$ on block $\ell+1$ pretending they are Carole and Daniel 
  receiving $x^{\ell +1},y^{\ell+1}$
  \\
  \label{r33} Alice sends Carole's memory state to Bob
  \\
  \label{r4} Bob sends Daniel's memory to Alice
  \\
  \label{r5} Alice and Bob on their own update Carole and Daniel memory until reaching indices in 
  $[\frac{n}{2} + \ell \alpha S, \frac{n}{2} + (\ell+1) \alpha S]$. At this point, they jointly simulate the protocol $\Pi$ again.
\\
\label{r6} Again Alice sends Daniel's memory to Bob.
  \\
  \label{r7} Bob knows and outputs $IP(x^{\ell+1},y^{\ell+1})$ (details below).
 \end{algorithm}
 
  We will now analyze the correctness and cost of $\Pi_{\ell+1}$.

  \emph{Cost :} Using \ref{sec:simclock}, we can already evaluate the cost to be in expecatation $R_{\ell+1} \times \left(S + \log(\alpha S) \right)$
  for parts \ref{r33} and \ref{r6} together and less than $20S$ for all the other memory exchanges
  at \ref{r3},\ref{r4},\ref{r5},\ref{r7}. The total cost in expectation
  is thus less than 
  \eq{
  R_{\ell+1} \times \left(S + \log(\alpha S) \right) + 20S.
  }
  \emph{Correctness :}
Ultimately, at line \ref{r6}, Alice sends back the memory of Daniel and Bob has
the final memory of Carole and Daniel, as if they each received $X'=X$ and $Y=Y'$, so he can answer
$f(XY,YX)=IP(X,Y)=IP(x^{\ell +1},y^{\ell+1})$.

The protocol $\Pi_{\ell+1}$ errs only if the initial procedure$\Pi$ errs.
Hence $\Pi_{\ell+1}$ solves $IP_{\alpha S}$ in the usual model of $CC$ with expected number of bits $\leq \frac{2\alpha S}{3}$, which is a contradiction.
\end{proof}

\section{The composition Theorem}
We start by exposing the idea of our general bound  by providing a tight tradeoff for communication rounds and memory in the SC model for the EQ function. Note that for the public coin randomized case EQ is a trivial problem. 
Then we show Theorem \ref{thm:genlem} in Section \ref{sec:genlem}.

\subsection{Deterministic EQ}\label{sec:deteq}

\begin{theorem}
Any deterministic protocol in the SC model computing $EQ_n$ with $R$ communication rounds and $S$ memory, must have $R \cdot S = \Omega(n)$.
\end{theorem}

\begin{proof}
We are going to use  a protocol for $EQ$ in the SC model with inputs of length $n$, to solve $EQ$ on smaller inputs
of length $k \leq n$, in the standard communication model. 

Assume we have a protocol $\Pi$ in the SC model with $R$ number of rounds.
Let $k := \frac{n}{R}$ and for simplicity we assume that $k$ is integer. With
this choice of $k$, for any pair of inputs, $X,Y$ in the original protocol there is a silence of length $k$,
by the pigeonhole principle. 
By \emph{a silence}
of length $k$, we mean a period of length $k$ where all messages sent are empty.
However, note that the position of this silent period might be different for different inputs.

We will describe a deterministic protocol for $EQ_k$ in the usual communication model, for some appropriate choice of $k$, using only $2S + \log k$ bits, which will imply $2S + \log k \geq k$ using
the bound for $EQ_k$. We can assume $k \geq S$, otherwise the desired bound holds, and for $k$ large enough
$k \leq 2.1 S$, which implies $R \cdot S = \Omega(n)$.

\RestyleAlgo{boxruled}
\LinesNumbered

 \begin{algorithm}[ht]
 \caption{Let $\Pi_k$ be the following protocol
for $EQ_k$, on inputs $x,y \in \{0,1\}^k$.}
Set $X := x^{\frac{n}{k}} \in \{0,1\}^n$, the concatenation of $x$, $\frac{n}{k}$ times. \\
Alice simulates $\Pi$ on $X,X$, without communication:
she simulates \emph{both sides}, call them Carole and Daniel,
assuming they \emph{both} have input $X$.\\
 She sends Bob the index $i \in [k]$ of $x$ she is reading one step after the beginning of the silent period.\\
\label{s21} She sends as well Daniel's initial and final state $\stt_{init},\stt_{finish}$.\\
Bob then outputs $1$ iff $(1)$ Daniel would also have remained silent if he was getting $y$ (shifted by $i$)
and $(2)$ starting at $\stt_{init}$ Daniel would have reached $\stt_{finish}$ by the end of the silent period.
\end{algorithm}
We now prove the correctness of the protocol. 

If $x=y$, then conditions $(1)$ and $(2)$ hold.
If 
$x \neq y$ then conditions $(1)$ and $(2)$ cannot hold. Assume they do, 
then we can build two inputs $X,Y$ to the original protocol
such that the answer on $(X,X)$ and $(X,Y)$ are the same, which is a contradiction to the correctness of $\Pi$. 

The word $X = x^{\frac{n}{k}}$ is chosen to be the same as above,
and $Y$ is equal to $X$
except on the first $k$ bits of the silent period where $x$ is replaced by $y$.

All together, we have $x=y$ iff conditions $(1)$ and $(2)$ are correct.
In other words, $\Pi_k$ is a valid protocol for $EQ_k$.
\end{proof}

\subsection{Proof of Theorem \ref{thm:genlem}}\label{sec:genlem}
In this section we prove Theorem \ref{thm:genlem}.
The bounds are proven in the same way in the deterministic/randomized setting and for bits or rounds. The difference between bits and rounds comes from the difference in the bounds in Proposition \ref{prop:clocksimth}. In what follows, we focus on lower bounding the rounds in the randomized setting and leave to the reader the slight modifications needed for the other cases. The modifications required for the case $G = \bigoplus$ are discussed in the end of the proof.

The idea of the proof is the following: we start by any protocol $\Pi$ for $f$ in the SC model and we demonstrate that there exist inputs $x,y$ that require a lot of communication rounds on average under $\Pi$.
 More precisely, we argue that there must exist inputs $x,y$ such that some communication happens
 on each of the $L$ blocks \emph{on average over the randomness}. This is achieved by deriving a protocol, denoted $\Pi_{\ell}$ in the usual communication model for $g_\ell$ with error $\leq \varepsilon + \delta$ and that uses less than (roughly) $R_{\ell}\cdot S$ bits, where $R_{\ell}$ is the expected number of rounds of communication under $\Pi$ at block $\ell$. This shows a lower bound on $R_{\ell}$. The protocol $\Pi_{\ell}$ is obtained by embedding inputs of $g_\ell$ in bigger inputs for $f$, and running the purported efficient protocol $\Pi$ for $f$.
 
 More precisely, let $\Pi$ be a protocol for $f$ in the SC model. Let $a,b$ be the strings given by the assumption that $G$ is non trivial. 

Our goal is to find inputs $x = x^1,\ldots , x^L$ and $y=y^1,\ldots, y^L$ for $f$ with some specific properties. To find these inputs, we first find $(x^1,y^1)$, then given the values of $(x^1,y^1)$ we find $(x^2,y^2)$ and so forth, until we complete the inputs for $f$.

For each $\ell \in [L-1]$, we will prove that given any so-far picked inputs $x^{\leq \ell}=x^1,\ldots,x^\ell$ and $ y^{\leq \ell}=y^1,\ldots,y^\ell$ that satisfy $g_1(x^1,y^1),\ldots,g_{\ell}(x^\ell,y^\ell) = a_{\leq \ell}$, we can always find $(x^{\ell+1},y^{\ell+1})$ with the following two properties.
First, $g_{\ell+1}(x^{\ell+1},y^{\ell+1}) = a_{\ell+1}$ for the string $a$ that witnesses $G$ is nontrivial. 
Note that $a_{\ell+1} \in \{0,1\}$ and there should always be inputs for which $g_{\ell+1}(x^{\ell+1},y^{\ell+1}) = a_{\ell+1}$, otherwise $g_{\ell+1}$ is constant and we can ignore it from the computation. 
Second, in the protocol $\Pi$ for $f$, the expected number of communication rounds on block $\ell+1$ is at least      $\frac{\delta C_{\varepsilon+\delta}\left(g_{\ell+1}\right) - \co S}{S + 2\log t_{\ell +1}+1}$.
In Lemma \ref{lem:lowerbmain} we prove that it is indeed possible to find such inputs for all $\ell$ and for any prefix. This concludes the proof of the theorem since we first use the lemma to pick $(x^1,y^1)$ with high expected communication, then we use it again to pick $(x^2,y^2)$ with high expected communication (given the prefix $(x^1,y^1)$) and so forth until we pick the entire inputs $x,y$ in a way that the bounds in the theorem hold. For the case when $G=\bigoplus$, we can remove the $\delta$ from the bound by using Lemma \ref{lem:lowerboplus}.

In the following lemma, observe that we ask for a lower bound on the expected (over public randomness $r$ in the protocol)
 number of rounds instead of the more standard \emph{maximum} number of rounds over $r$. This is because,
 given $\Pi$ and inputs $x,y$,
the strings $r$'s which give high communication can a priori be different on each block.
 What we are after is one fixed string $r$ for which the communication is high 
 on many blocks. Of course, proving our lower bound for the expected number of rounds also implies it for the worst case.
 
\begin{lemma}\label{lem:lowerbmain}
For any $\ell \in [L-1]$ and for any pair of prefixes $x^{\leq\ell}, y^{\leq \ell}$ such that $g_1(x^1,y^1),\ldots,g_{\ell}(x^\ell,y^\ell) = a_{\leq \ell}$ it is possible to find a pair of inputs for block $\ell+1$, $(x^{\ell+1},y^{\ell+1})$ such that
\begin{compactitem}
\item $g_{\ell +1}(x^{\ell+1},y^{\ell +1}) = a_{\ell +1}$;
\item The expected 
number of rounds of communication on block $\ell+1$ 
under $\Pi$, for inputs starting with $x^{\leq \ell+1}, y^{\leq \ell+1}$ is at least
$R_{\ell+1} := \frac{\delta C_{\varepsilon+\delta}\left(g_{\ell+1}\right)  - S}{S + 2\log t_{\ell +1}+1}$.
\end{compactitem}
\end{lemma}

\begin{proof}[Proof of Lemma \ref{lem:lowerbmain}]
Note that we have assumed without loss of generality that the functions $g_\ell$ are not constant. 
For any $\ell \in [L-1]$, denote by $b_\ell \in \{0,1\}^{L-\ell-1}$ the input such that $G(a_{\leq \ell}0b_\ell) \neq G(a_{\leq \ell}1b_\ell)$.

We now describe a protocol 
 $\Pi_{\ell+1}$ 
 for $g_{\ell +1}$ in the usual communication model that computes $g_{\ell +1}$ with error $\leq \varepsilon + \delta$ and worst case communication cost in bits
 \eql{
\embedcost := \frac{1}{\delta} \left( R_{\ell+1} \left(S+ 2\log t_{\ell+1} +1 \right) + S\right).}{eq:embedcost}
 But this quantity has to be greater than $C_{\varepsilon+\delta}(g_{\ell+1})$ and this implies the lower bound on $R_{\ell+1}$ given in the statement.
 Note that we need to describe a protocol $\Pi_{\ell+1}$ that works \emph{for all inputs} even though the guarantee we have for $\Pi$ on block $\ell+1$ is only for inputs $x^{\ell+1}, y^{\ell+1}$ such that $g_{\ell +1}(x^{\ell+1},y^{\ell +1}) = a_{\ell +1}$. 
 The protocol $\Pi_{\ell+1}$
 is described in pseudocode below. Let us call Alice and Bob the players that want to use $\Pi_{\ell+1}$ to compute $g_{\ell+1}$ and Carole and Daniel the fictitious players simulating $\Pi$. 
 
\RestyleAlgo{boxruled}
\LinesNumbered
 \begin{algorithm}[ht]
 \caption{The protocol $\Pi_{\ell+1}$ for computing $g_{\ell+1}$ on inputs $x^{\ell +1},y^{\ell+1}$ based on $\Pi$}

 Alice and Bob simulate on their own the SC protocol $\Pi$
  on the lexicographically first inputs $x^{\leq\ell}, y^{\leq \ell}$ such that $g_1(x^1,y^1),\ldots,g_{\ell}(x^\ell,y^\ell) = a_{\leq \ell}$ for players Carole and Daniel. \\
    \label{p1}
 Then, Alice and Bob jointly simulate $\Pi$ on block $\ell+1$ pretending they are Carole and Daniel 
  receiving $x^{\ell +1},y^{\ell+1}$ respectively. 
  \\\label{p2} 
Alice sends Carole's memory state to Bob.\\\label{p3} 
 Bob finishes the simulation on his own, pretending Carole and Daniel are
  getting as input some strings
  $x',y'$ such that $g_{>\ell +1}(x',y') = b_\ell$. Let $c$ be the output of the protocol $\Pi$. Bob outputs $c$ if $G(a_{\leq \ell}0b_\ell)=0$ and $1-c$ if $G(a_{\leq \ell}0b_\ell)=1$. \\
  \label{p4}
  If at any point, the exchanged bits reach $\embedcost$ (defined in Equation \ref{eq:embedcost}), Bob outputs $1-a_{\ell+1}$ and the protocol stops. \label{p5}
 \end{algorithm}
 
  We now analyze the protocol $\Pi_{\ell+1}$.
  Bob has all the data to simulate exactly the protocol $\Pi$, which tries to compute the function
$f(x,y)= G\left(a_{\leq \ell}g_{\ell}(x^{\ell+1},y^{\ell+1})b_{\ell}\right)$. Note that it holds that $G\left(a_{\leq \ell}g_{\ell}(x^{\ell+1},y^{\ell+1})b_{\ell}\right) = g_{\ell}(x^{\ell+1},y^{\ell+1})$ when $G(a_{\leq \ell}0b_\ell)=0 \neq G(a_{\leq \ell}1b_\ell)$ (this is why Bob outputs $c$), and $G\left(a_{\leq \ell}g_{\ell}(x^{\ell+1},y^{\ell+1})b_{\ell})\right) =1- g_{\ell}(x^{\ell+1},y^{\ell+1})$ when $G(a_{\leq \ell}0b_\ell)=1\neq G(a_{\leq \ell}1b_\ell)$ (this is why Bob outputs $1-c$).
  
It is clear that the communication cost of the protocol $\Pi_{\ell+1}$ is $\leq \embedcost$, since the protocol exits before at step $(6)$. 
All we need to do to conclude is justify that the error is less than $\varepsilon+ \delta$. There are two cases when $\Pi_{\ell+1}$ errs. First, when the protocol $ g_{\ell+1}\left(x^{\ell+1},y^{\ell+1}\right)=a_{\ell+1}$ and the protocol exits (and hence outputs $1-a_{\ell+1}$); second, when the protocol completes the simulation of $\Pi$ and $\Pi$ errs. We can upper bound the error probability as
\begin{eqnarray}\label{errr}
 \Pr \left[\Pi_{\ell+1} \mbox{ errs}\right] \leq 
 \Pr[\Pi_{\ell+1} \mbox{ exits  when } g_{\ell+1}(x^{\ell+1},y^{\ell+1})=a_{\ell+1}] + 
 \Pr[\Pi \mbox{ errs}].
\end{eqnarray}

We know that $\Pr[\Pi \mbox{ errs}] \leq \varepsilon$ and hence it remains to show that $\Pr[\Pi_{\ell+1} \mbox{ exits when }  g_{\ell+1}(x^{\ell+1},y^{\ell+1})=a_{\ell+1}] \leq \delta$.

First, recall that for all possible inputs $(x^{\ell+1},y^{\ell+1})$ with $g_{\ell +1}(x^{\ell+1},y^{\ell +1}) = a_{\ell +1}$, the expected  number of communication rounds on block $\ell+1$ in $\Pi$ is bounded by $R_{\ell+1}$.
  
Using the results in section \ref{sec:simclock}, the protocol $\Pi$ on block $\ell +1$ can be simulated by Alice and Bob in the usual communication model with a multiplicative overhead $S+2\log t_{\ell+1}+1$.
Hence, the expected communication cost of step \ref{p2} is $R_{\ell+1} \left(S+2\log t_{\ell+1}+1 \right) $. 

The communication cost at step \ref{p3} is equal to $S$, and so, the total expected cost of $\Pi_{\ell+1}$, when $g_{\ell +1}(x^{\ell+1},y^{\ell +1}) = a_{\ell +1}$, is 
smaller than
$R_{\ell+1}(S+2\log t_{\ell+1}+1) + S = \delta A_{\ell+1}$.
Thus the probability the protocol exits because the transcript reached $A_{\ell+1}$ in this case is smaller than $\delta$ by Markov inequality.

 Hence we established that $\Pi_{\ell+1}$ is a protocol for $g_{\ell+1}$ with error $\varepsilon+\delta$ using 
  $\leq \embedcost$ bits in the worst case over inputs.
  This concludes the proof.
 \end{proof}
%
%
%

\paragraph{The case $G = \bigoplus$}
In this case, the proof is similar but simpler. In fact, for the $\bigoplus$ function, we do not need to enforce that $g_{\ell}(x^{\ell},y^{\ell}) = a_{\ell}$ for any $a_\ell$, since basically all prefixes do \emph{not} fix the output of $G$. 
The analogous part to Lemma \ref{lem:lowerbmain} reads.
\begin{lemma}\label{lem:lowerboplus}
For any given SC protocol $\Pi$ for $f$, for any $\ell \in [L-1]$ and for any pair of prefixes $x^{\leq\ell}, y^{\leq \ell}$, it is possible to find a pair of inputs for block $\ell+1$, $(x^{\ell+1},y^{\ell+1})$ such that the expected 
number of rounds of communication on block $\ell+1$ in $\Pi$ is at least
$\frac{C^{avg}_{\varepsilon}\left(g_{\ell+1}\right)  - S}{S + 2\log t_{\ell +1}+1}$.
\end{lemma}
Given the lemma, the conclusion of Theorem \ref{thm:genlem} for the case of $G=\bigoplus$ follows easily.
\begin{proof}[Sketch of proof for Lemma \ref{lem:lowerboplus}]
Again we use $\Pi$ to derive  a protocol for $g_{\ell+1}$ in the usual communication model, using $R_{\ell+1} (S+ 2 \log t_{\ell+1} +1) +S$ bits in expectation over block $\ell+1$, and error $\varepsilon$.
This should be at least $C^{avg}_{\varepsilon}(g_{\ell+1})$ by definition of $C^{avg}_{\varepsilon}(\cdot)$ and the bound follows.
 Note that in this case, in step $4$, Bob can choose any suffix he wants to finish the simulation of $\Pi$. And since we look at the average case complexity, we do not need step $5$ and the protocol $\Pi_{\ell+1}$ can always finish the simulation of $\Pi$ on block $\ell+1$. Last, for the error, we simply have this time
\eq{
 \Pr \left[\Pi_{\ell+1} \mbox{ errs}\right] &= \Pr [\Pi \mbox{ errs}].
}
\end{proof}

\section{Removing overlapping edges}\label{sec:disjoint}
In our previous construction, edges sent to Alice and Bob may overlap.
We now explain how to transform a distribution where edges sent to Alice and Bob may overlap to a distribution where  edges are disjoint.
The transformation is built on a gadget, which replaces $G$ by two copies $G_{\times}$ or $G_{\parallel}$ defined as follows.

\begin{definition}[The graph $G_{\times}$ and the graph $G_{\parallel}$]
Starting from a graph $G = (P,Q,E)$ with edge set $E$ defined over vertices $P \times Q$ denote by $\ell : P, Q\rightarrow P_1,Q_1$ and $r : P,Q \rightarrow P_2,Q_2$ two bijections mapping $P,Q$ to disjoint copies $P_i,Q_i$, for $i = 1,2$.
We define the graph  $G_{\times}$ (resp. $ G_{\parallel}$)
over $(P_1 \sqcup P_2, Q_1 \sqcup Q_2)$
by putting the two edges $(\ell(x), r(y)), (r(x), \ell(y))$  (resp. $(\ell(x), \ell(y)), (r(x), r(y))$) for every  $(x,y) \in E$ (see Figure \ref{fig:parallel}).
\end{definition}

\begin{corollary}
We may assume in the reduction of Theorem \ref{thm:reduction} that the edges of both players are non overlapping.
\end{corollary}
\begin{proof}
Assuming we have a hard distribution $\mu$ for approximate matching in streaming, we are going to build a distribution $\mu_2'$ as in Definition \ref{def:mu2}, but with the extra constraint that the edges of both players are disjoint.
We denote by $\mu_2'$ the distribution where $G$ is sampled according to $\mu$, then Alice is given $\sigma(G_{\times})$ and Bob $\tau(G_{\parallel})$ where $\sigma, \tau$ are $G$-compatible (with the appropriate
modification to account for the doubling of the input set).

Note that the marginals of $\mu_2'$ have no overlapping edges by construction.
Let us explain at a high level why Theorem \ref{thm:reduction} still works using  $\mu_2'$ rather than $\mu_2$. 

Observe that  if $\mu$ is $(\alpha, n, m(n), \eta)$ hard then so is $\mu_{\times}$ (resp. $\mu_{\parallel}$) the distribution under the $\times$ (resp. $\parallel$) transform.
Hence 
Lemma \ref{lem:union} follows in the same way as before.
 On the other hand, the size of the input has doubled but so has the maximum matching.
 
%
\begin{figure}[!ht]
    \centering
    \includegraphics[width=0.4\textwidth]{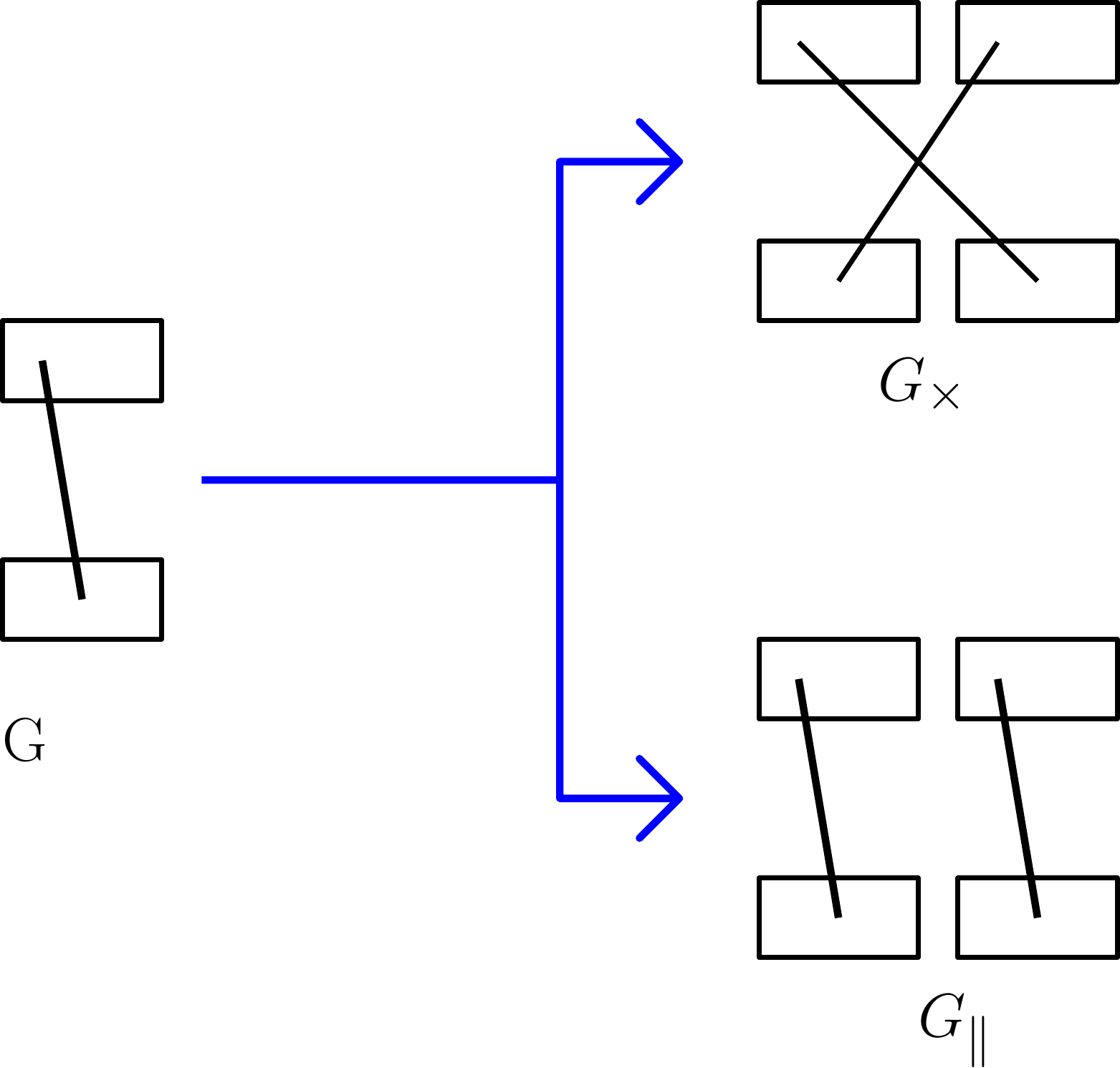}
        \caption{A bipartite graph $G$, and the associated graphs $G_{\times}$ and $G_{\parallel}$.} \label{fig:septree} \label{fig:parallel}
\end{figure}

\begin{figure}[!ht]
\vspace{-7cm}
\centering
    \includegraphics[width=0.6\textwidth]{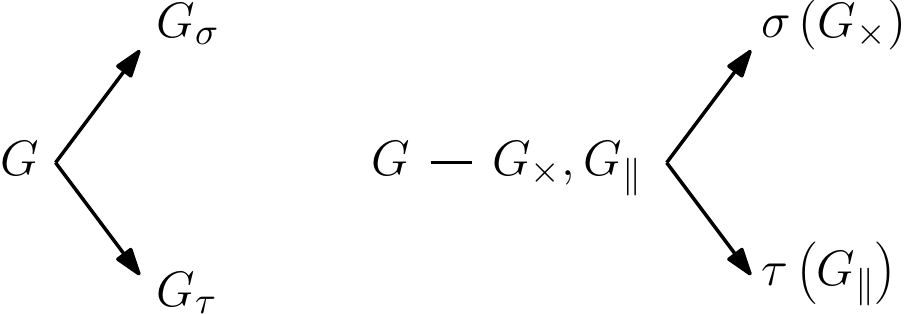}
        \caption{The difference between the two player distribution $\mu_2$ and $\mu_2'$ (no overlap).} \label{fig:septree} \label{fig:overlap}
\end{figure}

\end{proof}

\end{document}